\documentclass[letterpaper,11pt]{article}
\usepackage{amssymb}
\usepackage{algorithm}
\usepackage{algorithmic}
\usepackage{amsmath}
\usepackage[normalbibnotes,normallooseness,normalindent,normallists,normalbib,normalleading,normaltitle]{savetrees}

\usepackage{xspace}
\usepackage{sober}
\usepackage{mdwlist}
\usepackage{graphicx}

\newcommand{\old}[1]{}

\newcounter{algorithmLine}

\newtheorem{theorem}{Theorem}
\newtheorem{lemma}{Lemma}

\newtheorem{definition}{Definition}

\newcommand{\sq}{\hbox{\rlap{$\sqcap$}$\sqcup$}}
\newcommand{\qed}{\hspace*{\fill}\sq}
\newenvironment{proof}{\noindent {\bf Proof:}}{\qed\par\vskip 4mm\par}

\title{Breaking the $O(n^2)$ Bit Barrier: Scalable Byzantine agreement with an Adaptive Adversary}
\author{Valerie King \thanks{val@cs.uvic.ca; Department of Computer Science, University of Victoria, P.O.
Box 3055, Victoria, BC, Canada V8W 3P6.  This research was supported by an NSERC grant.} \and Jared Saia \thanks{saia@cs.unm.edu; Department of Computer Science, University of New Mexico, Albuquerque, NM 87131-1386. This research was partially supported by NSF CAREER Award 0644058, NSF CCR-0313160, and an AFOSR MURI grant.}}

\date{}                                           

\begin{document}
\maketitle

\thispagestyle{empty}

\begin{abstract}

We describe an algorithm for Byzantine agreement that is scalable in the sense that each processor sends only $\tilde{O}(\sqrt{n})$ bits, where $n$ is the total number of processors.  Our algorithm succeeds with high probability against an \emph{adaptive adversary}, which can take over processors at any time during the protocol, up to the point of taking over arbitrarily close to a $1/3$ fraction.  We assume synchronous communication but a \emph{rushing} adversary.  Moreover, our algorithm works in the presence of flooding: processors controlled by the adversary can send out any number of messages.  We assume the existence of private channels between all pairs of processors but make no other cryptographic assumptions.  Finally, our algorithm has latency that is polylogarithmic in $n$.  To the best of our knowledge, ours is the first algorithm to solve Byzantine agreement against an adaptive adversary, while requiring $o(n^{2})$ total bits of communication.
\end{abstract}

\ \\ \ \\ \ \\ \ \\ \ \\ \ \\ \ \\ \ \\ \ \\ \ \\ \ \\ \ \\ \ \\ \ \\ \ \\ \ \\ \ \\ \ \\ \ \\

This paper should \emph{not} be considered for the best student paper award.

\pagebreak
\setcounter{page}{1}

\section{Introduction}

Recent years have seen a rapid increase in the number of networks that are characterized by large sizes and little admission control.  Such networks are open to attacks by malicious users, who may subvert the network for their own gain.  To address this problem, the research community has been recently revisiting techniques for dealing with nodes under the control of a malicious adversary~\cite{kotla2007zyzzyva,clement-making,1529992,anderson2002worldwide}.

The Byzantine agreement problem, defined in 1982, is the \emph{sine qua non} of handling malicious nodes.   With a solution to Byzantine agreement, it is possible to create a network that is reliable, even when its components are not.  Without a solution, a compromised network cannot perform even the most basic computations reliably.  A testament to the continued importance of the problem is its appearance in modern domains such as sensor networks~\cite{shi2004designing}; mediation in game theory~\cite{ADGH,ADH}; grid computing~\cite{anderson2002worldwide}; peer-to-peer networks~\cite{rhea2003pond}; and cloud computing~\cite{wright2009contemporary}.  However, despite decades of work and thousands of papers, we still have no practical solution to Byzantine agreement for large networks.  One impediment to practicality is suggested by the following quotes from recent systems papers (see also~\cite{castro2002practical,Malkhi97unreliableintrusion,amir2006scaling,agbaria2003overcoming,1098025}):

\begin{itemize}
\item \emph{``Unfortunately, Byzantine agreement requires a number of messages quadratic in the number of participants, so it is infeasible for use in synchronizing a large number of replicas''}~\cite{rhea2003pond}

\item \emph{``Eventually batching cannot compensate for the quadratic number of messages [of Practical Byzantine Fault Tolerance (PBFT)]''}~\cite{CMLRS}

\item \emph{``The communication overhead of Byzantine Agreement is inherently large''}~\cite{Cheng2009219}

\end{itemize}

In this paper, we describe an algorithm for Byzantine agreement with only $\tilde{O}(n^{1/2})$ bit communication per processor overhead.  
Our  techniques also lead to solutions with $\tilde{O}(n^{1/2})$ bit complexity for universe reduction and a problem we call  the {\it  global coin subsequence problem},  generating a polylogarithmic length string, most of which are global coinflips generated uniformly and independently at random and agreed upon by all the good processors .  Our protocols are polylogarithmic in time and, succeed with high probability.\footnote{That is probability $1-1/n^{k}$ for any fixed $k$}

We overcome the lower bound of~\cite{DR} by allowing for a small probability of error.  This is necessary since this lower bound also implies that any randomized algorithm which always uses no more than $o(n^{2})$ messages must necessarily err with positive probability, since the adversary can guess the random coinflips
and achieve the lower bound if the guess is correct.

\subsection{Model and Problem Definition}
 We assume a fully connected network of $n$ processors, whose IDs are common knowledge.  Each processor has a private coin.  We assume that all communication channels are \emph{private} and that whenever a processor sends a message directly to another, the identity of the sender is known to the recipient, but we otherwise make no cryptographic assumptions.  We assume an {\it adaptive} adversary. That is, the adversary can take over processors at any point during the protocol up to the point of taking over up to a $1/3 - \epsilon$ fraction of the processors for any positive constant $\epsilon$.  The adversary is malicious: it  chooses the input bits of every processor,  bad processors can engage in any kind of deviations from the protocol, including false messages and collusion, or crash failures, while the remaining processors are good and follow the protocol. Bad processors can send {\it any} number of messages.

We assume a synchronous model of communication.  In particular, we assume there is a known upper bound on the transit time of any message and communication proceeds in rounds determined by this transit time.  The time complexity of our protocols are given in the number of rounds.  However, we assume a \emph{rushing} adversary that gets to control the order in which messages are delivered in each round.  In particular, the adversary can receive all the messages sent by good processors before sending out its own messages.

In the {\it  Byzantine agreement}  problem, each processor begins with either a 0 or 1. An execution of a protocol is {\it successful} if all processors  terminate and, upon termination, agree on a bit held by at least one good  processor at the start.  

The  {\it  global coin subsequence} $(s,t)$  problem  generates a string of length $s$ words, $t$ of  which are global coinflips generated uniformly and independently at random and agreed upon by all the good processors . We call $s$ an  {\it unreliable global coin sequence}.

\subsection{ Results}

We use the phrase {\it with high probability} ({\it w.h.p.}) to mean that an event happens with probability at least $1-1/n^c$
for every constant $c$ and sufficiently large $n$. For readability, we treat $\log n$ as an integer throughout. 

In all of our results, $n$ is the number of processors in a synchronous message passing model with an adaptive, rushing adversary that controls less than $1/3- \epsilon$ fraction of processors, for any positive constant $\epsilon$  We have three main results.  The first result makes use of the second and third ones, but these latter two results may be of independent interest. First, we show:

\begin{theorem} \label{t:main} {\sc [Byzantine agreement]}
There exists a protocol which w.h.p. computes Byzantine agreement, runs in polylogarithmic time,  and uses $\tilde{O}(n^{1/2})$ bits of communication.
\end{theorem}

Our second result concerns \emph{almost-everywhere} Byzantine agreement and,  \emph{almost-everywhere} global coin subsequence where a $(1-1/\log n)$ fraction of the good processors come to agreement on a good processor's input bit, or the random coin flip, resp.

\begin{theorem} \label{t:ae} {\sc [Almost Everywhere Byzantine agreement]}
For any $\epsilon >0$, there exists a protocol which w.h.p. computes almost -everywhere Byzantine agreement, runs in time $O((\log^{4+\epsilon}/\log \log n)$ and uses $\tilde{O}(n^{4/\epsilon})$ bits of communication per processor. In addition, this protocol can be used to solve an almost everywhere global coin subsequence $(s,2s/3) $ problem for an additional cost of $O(\log n/\log \log n)$ time and $\tilde{O}(n^{4/\epsilon})$ bits of communication per bit of $s$.   
\end{theorem}

Our third result is used as a subroutine of the previous protocol. 
\begin{theorem} \label{t:aebasparse}{\sc[Almost Everywhere Byzantine Agreement with Unreliable Global Coins}]
Let $S$ be a sequence of length  $s$ containing a subsequence of uniformly and independent random coinflips of length  $t$ known to 
$1-O(1/\log n)$  good processors.  Let $C_{1}$ and $C_{2}$ be any positive constants.  Then there is a protocol which runs in time $O(s )$ with bit complexity $O(\log n)$ such that with probability at least $1-e^{-C_{1}n} + 1/2^{t}$, all but $C_{2}n / \log n$ of the good processors commit to the same vote $b$, where $b$ was the input of at least one good processor.   
\end{theorem}

Our final result concerns going from almost-everywhere Byzantine agreement to everywhere Byzantine agreement.  It makes use of a simple consequence of our first result which is that not only can almost all of the processors reach agreement on a bit, but they can also generate a random bit.  We actually prove a result below that is stronger than what is necessary to establish Theorem~\ref{t:main}.

\begin{theorem}\label{t:ae2e}
Assume $n/2 + \epsilon n$ good processors agree on a message $M$ and there is an oracle which can generate each bit of a global coin subsequence $(s,t)$ in $O(1)$ time where $t > c\log n$. Then  there is a protocol that ensures with probability $1-1/n^c$  that all good processors output $M$ and n.  Moreover this protocol runs in $O(s \log n)$ time and uses $\tilde{O}( s n^{1/2})$ bits of communication per processor.
\end{theorem}

\subsection{Techniques}

Our protocol uses a sparse network construction and tournament tree similar to the network and tournament network in \cite{KSSV}.  This past result gives a bandwidth efficient Byzantine agreement algorithm for a \emph{non-adaptive adversary},  which must take over all its processors at the start of the algorithm.  The basic idea of the algorithm from~\cite{KSSV} is that processors compete in local elections in a tournament network, where the winners advance to the next highest level, until finally a small set is elected that is representative in the sense that the fraction of bad processors in this set is not much more than the fraction of bad processors in the general population.

This election approach is \emph{prima facie} impossible with an adaptive adversary, which can simply wait until a small set is elected and then can take over all processors in that set.  To avoid this problem, we make use of two novel techniques.  First, instead of electing processors, we elect \emph{arrays} of random numbers, each generated initially by a processor.  Second, we use secret sharing on these arrays to make sure that 1) the arrays are split among increasingly larger numbers of processors as the array is elected higher up in the tournament; and 2) the secrets in the arrays cannot be reconstructed except at the appropriate time in the protocol.  Critical to our approach is the need to iteratively reapply secret sharing on shares of secrets that were computed previously, in order to increase the number of shares when necessary in the protocol.

Another contribution of this paper is the algorithm we use to run an election.  In~\cite{KSSV}, elections were run by participants.  These elections used Feige's bin selection protocol~\cite{Feige} and a Byzantine agreement algorithm run among the small group of participants to agree on the bin selected by everyone.  Because we are now faced with an adaptive adversary, this approach fails.  In particular, we must now have a much larger sets of processors which come to agreement on the bins selected in Feige's protocol.  To achieve this, we make use of Rabin's algorithm~\cite{rabin1983randomized} run on a sparse network.  
\old{To ensure that Rabin's algorithm is bandwidth efficient, not only do we run it on a sparse network, but we also}
To run Rabin's algorithm, we supply it with an almost everywhere global coin sequence, where coinflips are generated from the arrays described above.  

Our final new technique is a simple but not obvious protocol for going from almost-everywhere Byzantine agreement and the global coin subsequence problem to everywhere Byzantine agreement with an adaptive adversary.  A past result~\cite{KSDISC09} shows that it is possible to do this with a non-adaptive adversary, even without private channels.  However, the technique presented in this paper for solving the problem with an adaptive adversary is significantly different than the approach from~\cite{KSDISC09}.

 \medskip
 
In Section~\ref{s:aeprotocol}, we describe the almost everywhere Byzantine agreement  and global coin subsequence protocols. The scalable version of Rabin's algorithm is in Section \ref{AEBACC}.  In Section~\ref{s:AE2E}, we describe the almost everywhere to everywhere protocol. 

\section{Related work}

As mentioned previously, this paper builds on a result from~\cite{KSSV} that gives a polylogarithmic time protocol with polylogarithmic bits of communication per processor for almost everywhere Byzantine agreement, leader election, and universe reduction in the synchronous full information message passing model with a \emph{nonadaptive} rushing adversary. 
 
Almost everywhere agreement in sparse networks has been studied since 1986. See \cite{KSSV,KSSV2} for references. The problem of almost everywhere agreement for secure multiparty computation on a partially connected network was defined and solved in 2008 in \cite{GO}, albeit with $\Omega(n^{2})$ message cost.

In \cite{KSSV2}, the authors give a sparse network implementation of their protocols from \cite{KSSV}.   It is easy to see that everywhere agreement is impossible in a sparse network where the number of faulty processors $t$ is sufficient to surround  a good processor.  A protocol in which processors use $o(n)$ bits may seem as vulnerable to being isolated as in a sparse network, but the difference is that without access to private random bits, the adversary can't anticipate at the start of the protocol where communication will occur.   In \cite{HKK}, it is shown that even with private channels, if a processor must pre-specify the set of processors it is willing to listen to at the start of a round,  where its choice in each round can depend on the outcome of its random coin tosses, at least one processor must send $\Omega(n^{1/3})$ messages to compute Byzantine agreement with probability at least $1/2+ 1/\log n$.  Hence the only hope for a protocol where every processor sends $o(n^{1/3})$ messages is to design outside this constraint.  Note that the Almost Everywhere Byzantine Agreement protocol falls within this restrictive model, but the Almost Everywhere to Everywhere protocol does not, as the decision of whether a message is listened to  (or acted upon) depends on how many messages carrying a certain value are received so far. 

\section{Almost everywhere protocol} \label{s:aeprotocol}

We first outline the protocol.  We label the processors $p_{1}, p_{2}, \ldots, p_{n}$.  The processors are arranged into nodes in a $q$-ary tree. Each processor appears in polylogarithmic places in each level of the tree, in a manner that will be described below. The levels of the tree are numbered from the leaf nodes (level 1) to the root (level $\ell^*$).  In addition, each processor, $p_{i}$, generates an \emph{array} of random bits, consisting of one \emph{block} for each level of the network and secret shares this with the processors in the $i^{th}$ node on level 1. 

Each node in the tree runs an election among $r$ arrays whereby a subset of $w$ arrays are selected.  In order to run this election at level $\ell$,
the $\ell$ block of each array supplies a random bin choice and random bits to run almost everywhere Byzantine agreement with common global coins to agree on each bin choice of every competing array.  It suffices that some of these coins are random and known almost everywhere. The  shares of  the remaining blocks of arrays  which remain in the competition are further subdivided into more shares and sent to the parent  (and erased from the current processors'
memories.) In this way, the more important the arrays, the more processors  need to be taken over to prevent its correct operation.

 Random bits are revealed as needed by sending the iterated shares of secrets down to {\it all}  the leaves of the subtree rooted where the election is taking place, collecting $\ell$-shares at each level $\ell-1$ to reconstruct $\ell-1$ shares.  In the level 1 nodes, each processor sends the other processors its share.

The winning arrays of a node's election compete in elections at the next higher level.
At the root there are a small number of arrays left to run almost everywhere Byzantine Agreement with a global coin.

The method of secret sharing and iterative secret sharing is described in Section \ref{secretsharing}.
Networks and communication protocols are described in Section \ref{network}; the election routine is described in Section \ref{election}.
The procedure for running almost everywhere Byzantine Agreement with unreliable coins is described in Section~\ref{AEBACC}.  The main procedure for almost everywhere Byzantine agreement is in \ref{main}.  The extension of the almost everywhere Byzantine Agreement protocol to a solution for  the global coin subsequence problem is in Section~\ref{GCS}.  Finally the analysis and correctness proof can be found in Sections \ref{analysis} and \ref{correct}, respectively.

\subsection{Secret sharing} \label{secretsharing}

We assume any (non-verifiable) secret sharing scheme which is a $(n, t+1)$  threshold scheme. That is, each of $n$ players are given shares of size
proportional to the message $M$ and $t+1$ shares are required to reconstruct $M$. Every message which is the size of $M$ is consistent with any subset of $t$ or fewer shares, so no information as to the contents of $M$ is gained about the secret from holding fewer than $t+1$ shares. See \cite{crypto} for details on constructing such a scheme.  We will make extensive use of the following definition.

\begin{definition}
{\it secretShare(s)}:   To share a secret sequence of words $s$ with $n_1$ processes (including itself)  of which $t_1$ may be corrupt, a processor (dealer) creates and distributes shares of each of the words using a $(n_1, t_1+1)$ secret sharing mechanism. 
Note that if a processor knows a share of a secret, it can treat that share as a secret. To share that share with  $n_2$ processors of which at most $t_2$ processors are corrupt,  it creates and  distributes shares of the share
using a $(n_2, t_2+1)$ mechanism and deletes its original share from memory. This can be iterated many times.
We define a {\it 1-share} of a secret to be a share of a secret and an {\it $i$-share} of a secret to be a share of an $i-1$-share of a secret.
\end{definition}
To reveal a secret sequence $s$, all processors which receive a share of $s$ from a dealer sends this shares to a processor  $p$ which computes the secret. This also may be iterated to first reconstruct  $i-1$ shares from $i$ shares, etc.,  and eventually the secret sequence. 
In this paper we assume secret sharing schemes with $t=n/2$. (This is quite robust, as any $t \in [1/3,2/3]$ would work.)

\begin{lemma} \label{secret}  If a secret  is shared in this manner  up to $i$ iterations, then an adversary which possesses  $t_i$ shares of each $i$-share learns no information about the secret.  \end{lemma}
\begin{proof}
The proof is by induction. For level 1, it is true by definition of secret sharing.
Suppose it is true up to $i$ iterations. 

 Let $v$ be any value. By induction, it is consistent
with the  known $t_j$ shares  on all levels $j \leq i$ and some assignment $S_i$ of values to sets of unknown $n_i-t_{i}$ $i$-shares. Then
consider the shares of these shares that have been spread to level $i+1$.  For each  value of an $i$-share given by $S_i$,  there is an assignment $S_{i+1}$ of values to the unknown  $n_{i+1}-t_{i+1}$ shares  consistent with the $t_{i+1}$ $i+1$- known shares.  Hence knowing in addition the
$t_{i+1}$ $i+1$-shares of each
$i$-share does not reveal any information about  the secret. 
\end{proof}

\subsection{Network and Communication}\label{network}
We first describe the topology of the network and then the communications protoocols. 

\subsubsection{Samplers}
Key to the construction of the network is the definition of an averaging sampler which was also used heavily in \cite{KKKSS-TALG,KSSV2}.
We repeat the definition here for convenience. 

Our protocols rely on the use of averaging (or
oblivious) samplers, families of bipartite graphs which define
subsets of elements such that all but a small number contain at
most a fraction of ``bad'' elements close to the fraction of bad
elements of the entire set.  We assume either a nonuniform model in
which each processor has a copy of the required samplers for a
given input size, or else that each processor initializes by
constructing the required samplers in exponential time.

\begin{definition}
Let ,$[r]$ denote the set of integers $\{1, \ldots, r\}$, and
$[s]^{d}$
 the multisets of size $d$ consisting of elements of $[s]$.
Let $H: [r] \rightarrow [s]^{d}$ be a function assigning
multisets 
of size $d$ to integers. We define the intersection of a multiset $A$ and a set $B$ to be the number of 
elements of $A$ which are in $B$.

$H$ is a $(\theta, \delta)$ sampler if
for every set $S \subset [s]$
at most a $\delta$ fraction of all inputs $x$ have
$\frac{|H(x) \cap S|}{d} > \frac{|S|}{s} + \theta$.
\end{definition}

\medskip

The following lemma establishing the existence of samplers
can be shown using the probabilistic method.  For $s'\in[s]$, 
let $deg(s')=|r'  \in[r] ~|~s.t.~ s \in H(r') \}|$.  A slight modification of Lemma 2 in \cite{KKKSS-TALG} yields:

\begin{lemma} \label{l:samp2}
For every $r,s,d, \theta, \delta > 0$ 
such that $2 \log_2(e) \cdot d \theta^2 \delta> s/r + 1-\delta$,
there exists a $(\theta, \delta)$
sampler $H: [r] \rightarrow [s]^{d}$ and for all $s \in [s]$, $deg(s) < O((rd/s) \log n)$.
 \end{lemma}

\medskip

For this paper we will use the term {\it sampler} to refer to a  $(1/\log n, 1/\log n)$ sampler, where $d=O( (s/r+1) \log^3 n)$.

\subsubsection{Network structure}

Let $P$ be the set of all $n$ processors. 
The network is structured as a complete $q$-ary  tree. The level 1 nodes (leaves) contain $k_1=\log^3 n$ processors.  Each node at height $\ell >1$ contains $k_{\ell}=q^{\ell} k_1$ processors; there are $(n/ k_{\ell}) \log^3 n$ nodes on level $\ell$; and the root node at height $\ell^*=\log_q (n/k_1)$ contains all the processors.  There are $n$ leaves, each assigned to a different processor.  The contents of each node on level $\ell$ is determined by  a sampler where $[r] $ is the set of nodes, $[s]= P$ and $d =k_{\ell}$. 

The edges in the network are of three types:
\begin{enumerate}
\item
{\it Uplinks:}   The {\it uplinks}  from  processors in a child node on level $\ell$  to processors in a parent node on level $\ell+1$ are determined by a sampler
of degree  $d=q \log^3  n$, $[r]$ is the set of processors in the child node and $[s]$ is the set of processors in the parent node.  
\item
$ \ell-links$:  The $  \ell- links$  between  processors in a node $C$ at level $\ell$ to $C$'s descendants at level $1$ is determined by a sampler with $[r]$ the set of processor in the node $C$, $[s]$ $C$'s level 1 descendants, and $d$ a subset of size $O(\log^3 n)$.  Here, $r=q^{\ell} k_1$; $s=q^{\ell}$; $d=O(\log^3 n)$
and the maximum number of $\ell-links$ incident to  a level 1 node is $O(k_1 \log^4n)$.

\item
{\it Links between processors in a node} are also determined by a sampler of polylogarithmic degree. These are described in the Almost Everywhere Byzantine Agreement with Global Coin protocol. 
\end{enumerate}

From the properties of samplers,  we have: 
\begin{enumerate}
\item Fewer than a $1/\log n$ fraction of the nodes on any level contain less than a $2/3 + \epsilon/2$ fraction of good processors (we call such nodes {\it bad nodes}). 
\item  There are fewer than a $1/\log n$ fraction of processors in every node whose uplinks are connected to fewer than a $2/3 +\epsilon -1/\log n$ fraction of good processors, unless the parent or child, resp. is a bad node.
\item
There are fewer than a $1/\log n$ fraction of processors in a node which are connected  through $\ell-links$ to a majority of bad nodes on level 1, in any subtree which has fewer than a $1/2 =\epsilon,$ fraction of bad level 1 nodes. 
\end{enumerate}

\subsubsection{Communication protocols}

We use the following three subroutines for communication.  Initially each processor $p_i$ shares its secret with all the processors in the $i^{th}$ node at level 1.
 
\smallskip
\noindent
$sendSecretUp(s)$: To  send up a secret sequence $s$, a  processor in a node uses $secretShare(s )$  to  send to each of its neighbors in its parent node (those connected by $uplinks$) a share of $s$.  Then the processor erases $s$ from its own memory.

\smallskip
\noindent
$sendDown(s,i)$: After a secret sequence has been passed up a path to a node $C$, the secret sequence is passed down to  the processors in the 1-nodes in the subtree.  To send a secret sequence $s$ down the tree, each processor in a node $C$ on level $i$ sends its $i$-shares of  $s$ {\it down} the uplinks it came from {\it plus the  corresponding uplinks from each of its other children}.  The processors  on level $i-1$ receiving the $i$-shares use these shares to reconstruct  $i-1$-shares of $s$. This is repeated for lower levels until  all the 2-shares are received by the processors in all the level 1 nodes in  $C$'s subtree. The processors in the 1-node  each send each other all their shares and reconstruct the secrets received. Note that a processor may have received an $i$-share generated from more than one $i-1$ share because of the overlapping of sets  (of uplinks) in the sampler.

\smallskip
\noindent
$sendOpen(s, \ell ):$ This procedure is used by a node $C$ on any level $\ell$  to learn a sequence $s$ held by the set of level 1 nodes in $C$'s subtree.  Each processor in the level 1 node  $A$ sends $s$ up the $\ell-links$ from $A$ to a subset of processors in $C$. 
 A processor in $C$ receiving $s$ from each of the processors in a level 1 node takes a majority to determine the node's  version of  $s$.
Then it takes a majority over  the values obtained from each of the level 1 nodes it is linked to.  

\subsubsection{Correctness of communications}\label{s:correct-comm}

\begin{definition} A {\it good node} is a node with at least $2/3 +\epsilon/2$ fraction of good processors and a bad node is a node which is not good.
(Note that for the lemma below, it suffices that that a good node contain a $1/2+ \epsilon$ fraction of good processors)
A {\it good path} up the tree is a path from leaf to root which has no nodes which become bad during the protocol. 
\end{definition}

\begin{lemma}
\enumerate
\item If  $sendSecretUp(s)$ is executed up a path in the tree and if the adversary learns the secret $s$, there must be at least one bad node on that path. 
\item Assume that $s$ is generated by a good processor and $sendSecretUp(s)$ is executed up a good path in a tree to a node $A$ on level $\ell$, followed by $sendDown(s,\ell)$  and then $sendOpen(s)$.  Further assume there are at least a $1/2 + \epsilon $ fraction of nodes among $A$'s descendants on level 1 which are good, and whose paths to $A$ are good.  Then a $ 1-1/\log n$ fraction of the good processors  in $A$ learn $s$.
\end{lemma}

\begin{proof}
In the protocol any secret shared to a good node on level 1 remains hidden from the adversary which receives no more than a $1/3-\epsilon$ fraction of the shares.  If it is passed to a good node on level 2, then since the uplinks are determined by a $(1/\log n, 1/\log n)$-sampler, no more than $1/\log n$ fraction of the uplink sets contain more than $1/3$ fraction of bad processors. Hence   
no more than an additional $1/\log n$ fraction of the 1-shares are revealed because the adversary has too many 2-shares.  Similarly, no more than 
an additional $1/\log n$ fraction of the 2-shares are revealed because the adversary has too many 3-shares. Hence if the secret is passed up a good path,
the adversary does not gain more than $\ell^*/\log n$ additional shares of the secret, or $O(1/\log\log n)$ of the shares, for a total less than $1/3-\epsilon/2$ fraction. Thus by Lemma \ref{secret} the adversary has no knowledge of any secret that is sent up a good path until that secret is released.

We consider a secret released  when it is first sent down from a node $A$. A secret will be reconstructed  by a processor when it is passed down good paths along the uplinks as  $2/3$ of all its shares are returned down the good paths to the leaves. If there are at least $1/2 +\epsilon$ fraction of level 1 nodes which are good and whose paths to $A$ are good, then a $1-1/\log n$ fraction of the good processors in $A$ have $\ell - links$ from a majority of 1-level nodes which have received the correct sequence. 
\end{proof}


\subsection{Election}\label{election}
Here we describe Feige's election procedure \cite{Feige}, adapted to this context.  We assume $r$ candidates are competing in the election.  The election algorithm is given below.

\begin{definition}
Let $numBins = r/ (5c \log^3 n)$, and let a \emph{word} consist of $\log numBins$ bits.  A {\it block} $B$ is a sequence of bits, beginning with an initial word (bin choice) $B(0)$ followed by $r$ words $B(1), B(2),..,B(r) $, which will be used as coins in running Byzantine agreement on each bit of the bin choices for each of the $r$ candidates.  The input to an election is a set of $r$ candidate blocks labelled $B_1,...,B_r$. The output is 
a set of $r/numBins$ indices $W$.  Let $w=|W|=5c \log^3 n$.
  
 \end{definition}

\begin{algorithm}\label{a:ep}
\caption{Election Protocol}
\begin{enumerate}
\item In parallel, for $i=1,...,r$, the processors run almost everywhere Byzantine agreement on the bin choice of each of the $r$ candidate blocks.  Round  $j$ of the Byzantine agreement protocol to determine  $i$'s bin choice  is run  using the $i^{th}$ word of the  $j^{th}$ processor's block $B_j(i)$.   Let $b_1,...,b_r$ be the decided bin choices.
\item
Let  $min  =\min\{i  ~|~  \sum_j B_j(0)=i \}$. Then   $W \leftarrow \{j~|~ B_j(0)=min\}$.
If $|W| < r/numBins$ then $W$ is augmented by adding the first  $r/numBins -|W|$ indices that would otherwise be omitted. 
\end{enumerate} 
\end{algorithm}

If we assume that the bin choices are agreed upon by all processors then Feige's result for the atomic broadcast model holds:

\begin{lemma} \cite{Feige} \label{Feige}
Let $S$ be the set of bin choices generated independently at random. Then even if the adversary sets the remaining bits after seeing the bin choices of $S$, with probability at least $1- 2^{- \epsilon^2 |S|/(3 numBins )}$  there are at least $(1/numBins - \epsilon)|S| $   winners from $S$.  E.g., if  $|S| > 2/3 r$ and
$r/numBins > 5c \log^3  n$  then with probability $1-1/n^c$ the fraction of winners from $S$ is at least $|S|/r  - 1/\log  n$. 
\end{lemma}

\subsection{Main protocol for Almost Everywhere Byzantine Agreement}\label{main}

The main protocol for Almost Everywhere Byzantine Agreement is given as Algorithm 2.  Figure 1, which we now describe, outlines the main ideas behind the algorithm.  The left part of Figure 1 illustrates the algorithm when run on a 3-ary network tree.  The processors are represented with the numbers $1$ through $9$ and the ovals represent the nodes of the network, where a link between a pair of nodes illustrates a parent child relationship.  The numbers in the bottom part of each node are the processors contained in that nodes.  Note that the size of these sets increase as we go up tree.  Further note that each processor is contained in more than one node at a given level.  The numbers in the top part of each node represent the processors whose blocks are candidates at that node.  Note that the size of this set remains constant ($3$) as we go from level 2 to level 3.  Further note that each processor is a candidate in at most one node at a given level.

The right part of Figure 1 illustrates communication in Algorithm 2 for an election occurring at a fixed node at level $\ell$.  Time moves from left to right in this figure and the levels of the network are increasing from bottom to top.  Salient points in this figure are as follows.   First, bin choices are revealed by (1) communication in the $sendDown$ protocol that moves hop by hop from level $\ell$ down to level $1$ in the network and (2) communication in the $sendOpen$ protocol that proceeds directly from the level $1$ leaf nodes to the level $\ell$ nodes.  Second, Byzantine agreement occurs at level $\ell$, via communication between the processors in the node at level $\ell$ \emph{and} communication down and up the network in order to expose the coins, one after another, as needed in the course of the Byzantine agreement algorithm.  Finally, shares of the blocks of the winners of the election at the node at level $\ell$ are sent up to the parent node at level $\ell+1$.

\begin{figure}[t]
\begin{center}
\begin{tabular}{ll}
\includegraphics[scale=0.32]{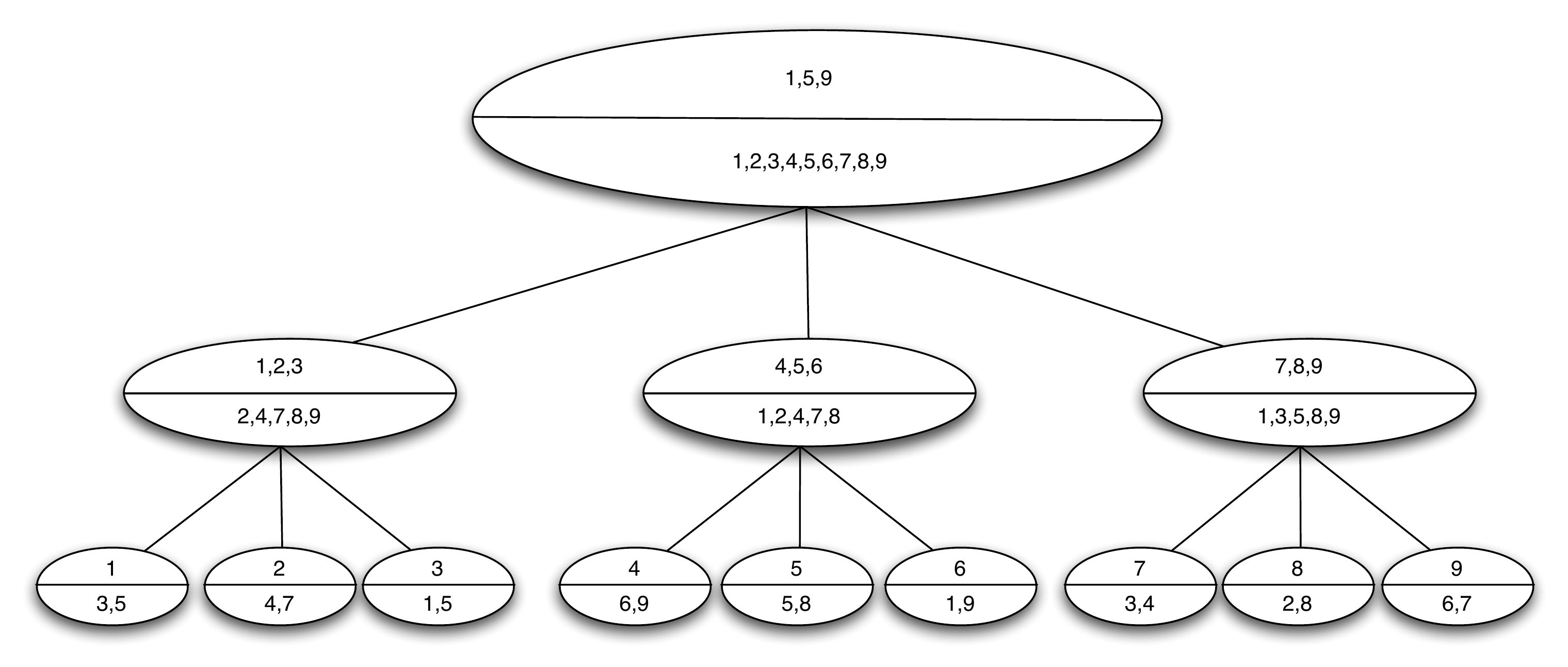} &
\includegraphics[scale=0.32]{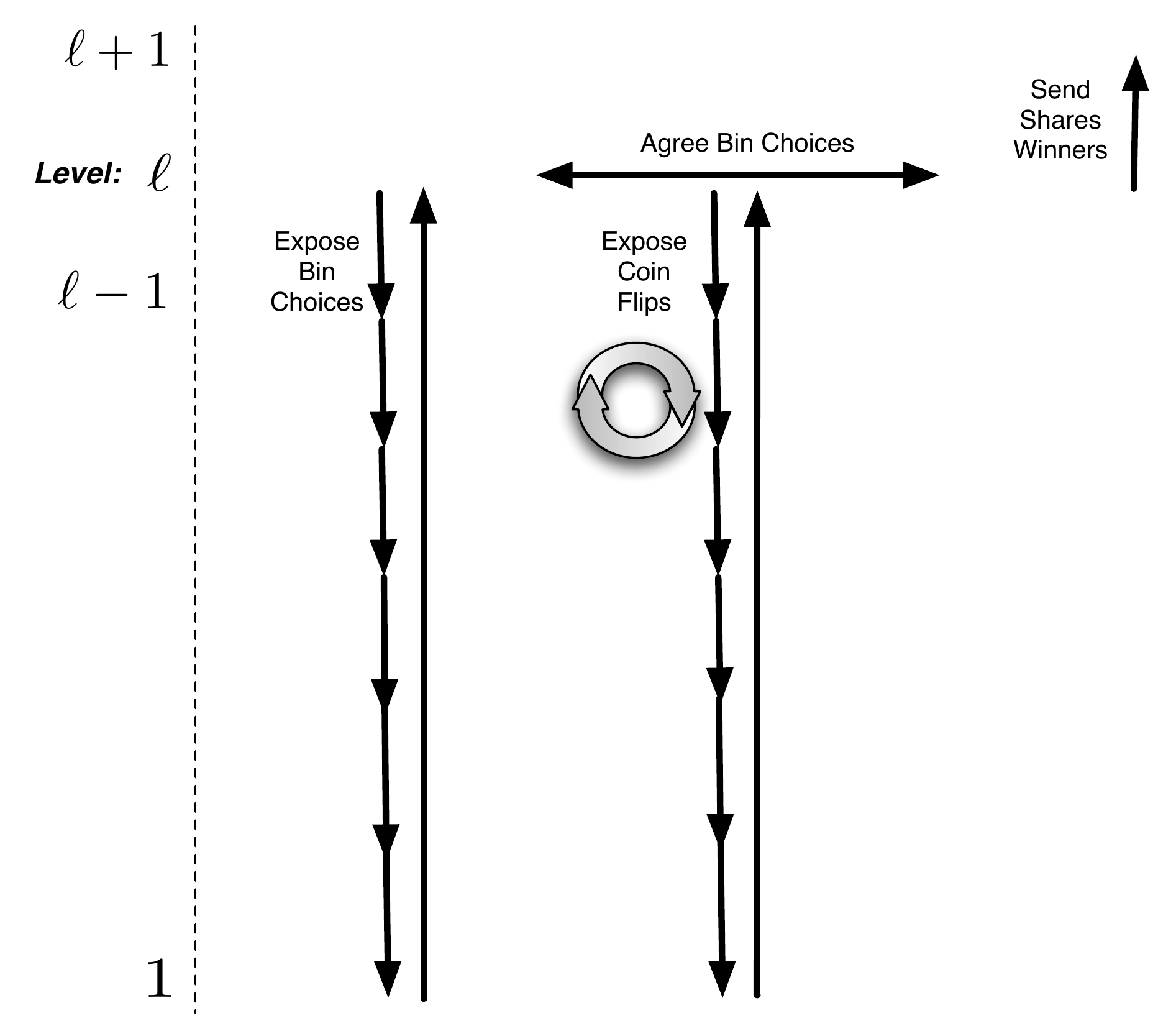}
\end{tabular}
\caption{Left: Example run of Algorithm 2 on a small tree; Right: Communication in different phases of Algorithm 2 for a fixed level $\ell$. }
\end{center}
\label{f:aeBA}
\end{figure}

\begin{algorithm} \label{aeBA}
\caption {Almost Everywhere Byzantine Agreement}
\begin{enumerate}
\item
For all $i$ in parallel
\begin{enumerate}
\item
 Each processor $p_i$ generates an array of  $\ell $ blocks  $B_i$ and  uses $secretshare$ to share its array with the $i^{th}$ level 1 node;
\item
Each processor in the $i^{th}$  level 1 node uses  $sendSecretUp$ to share its 1-share of $B_i$ with its parent node
and then erases its shares from memory.  \\
\end{enumerate}
\item
Repeat for $\ell=2$ to $\ell^*-1$\\
\begin{enumerate}
\item
For each processor in each node $C$ on level $\ell$:\\
   for $t=1,...,w$ and $i=1,...,q-1$, let  $B_{(i-1)w+t}$ be the $t^{th}$ array sent up from child $i$. 
( If $\ell =2  $ then $w=1$ ) \\
$W \leftarrow B_1 |  \! | B_2 | \! | ...| \! | B_{r}$  \\

Let $F$ be the sequence of first blocks of the arrays of $W$, i.e., the $i^{th}$ array of $F$ is the first block of $i^{th}$ array of $W$.  Let $S$ be the sequence of the remaining blocks  of each array of $W$.\\

 {\bf Expose bin choices:}\\
In parallel,  for all candidates $i=1, 2,..,r$

 \begin{enumerate}
\item
 $sendDown(F_i(1))$;
  \item
 $sendOpen(F_i(1), \ell)$. \\
\end{enumerate} 

\item {\bf Agree on bin choices:}\\
If $\ell <\ell^*$ then for rounds  $i=1,...,r$ \\
\begin{enumerate}
\item
{\bf Expose coin flips}: Generate $r$ coinflips for the $i^{th}$ round of Byzantine agreement to decide each of $r$ bin choices.\\
In parallel,  for all contestants $j=1,..,r$ 
\begin{enumerate}
\item
$sendDown(F_i(j))$; upon receiving all $1$-shares, level 1 processors compute the secret bits $F_i(j))$;
\item
 $sendOpen(F_i(j), \ell)$.
\end{enumerate}
\item
Run the $i^{th}$ round of a.e. Byzantine agreement in parallel to decide the bin choice of all contestants. \\
\end{enumerate}

\item {\bf Send Shares of Winners}: 
Let $W$ be the winners of the election decided from the previous step (the lightest bin).  Let $S'$ be the subsequence of $S$ from $W$;   All processors in a node at level $\ell$ use $sendSecretUp(S')$ to  send $S'$ to its parent node
and erase $S'$ from memory. \\
   \\
\end{enumerate}
\item
All  processes in the single node on level $\ell^*$ run a.e. Byzantine agreement once using their initial inputs  as inputs to the protocol  (instead of bin choices) and the remaining block of each contestant. (Note that only two bits of  this block are needed.) \\
For rounds $i = 1, 2,..., qw$, 

 \begin{enumerate}
\item
 $sendDown(F_i(2),i)$;

 \item
 $sendOpen(F_i(2), \ell)$. 
\item
Use $F_i(2)$ to run the $i^{th}$ round of a.e. Byzantine agreement. 
\end{enumerate}

\end{enumerate}
\end{algorithm}

\subsection{Modification to output a sequence of mostly random bits}\label{GCS}

The Almost Everywhere Byzantine Agreement can be modified easily to solve the global coin subsequence $(s,2S/3)$ for $s$ a sequence of $wq$ words.  Add one more block of the desired length to each processor's array at the start. At level $\ell^*$, use sendDown and sendOpen to recover each word, one from each of the $wq$ contestants. The time and bit complexity is given in Theorem \ref{t:ae}. 

\subsection{Bit complexity and running time analysis}\label{analysis}

The proof of the following lemma is given in Section~\ref{s:pf-analysis}

\begin{lemma}\label{l:analysis}
For any $\delta >0$, Almost Everywhere Byzantine Agreement protocol runs in $\tilde{O}(n^{4/\delta}) $ bits  and time $O((\log n)^{4 + \delta} / \log \log n) )$.
\end{lemma}

\subsection{Proof of correctness}\label{correct}

An election is {\it good}  at a node $A$  if the processors can carry out  a.e.\ Byzantine agreement and a $1-1/\log n$ fraction of good processors agree on the result. Recall that the a.e. Byzantine agreement with common coins protocol succeeds w.h.p., if $2/3+\epsilon $ fraction of processors in the node are good
and bits can be generated so that at least $c\log n$ are random and for each of these, there is a fraction of $1-1/\log n$ fraction of good processors which agree on it. 
Then an election is good if  (1) $A$ is a good node; (2) at least $c\log n$ contestants are good processors $p_i$ with good paths from 
the assigned $i^{th}$ level 1 node to  $A$ so that that secrets are correctly transmitted up the tree without the adversary learning the secret until it it released; and (3) there must be a $1/2 + \epsilon$ fraction of level 1 nodes in $A$'s subtree which have good paths to $A$, so that a $1-1/\log n$ fraction of good processors in $A$ learn the random bin selections and random bits of the good arrays that are competing.  Finally, if fewer than $c \log^3 n$ good arrays compete in an election, the probability of correctness of the election is diminished, see Section~\ref{election}. 
  
 Condition (3) is sufficient to show that  between the time the secret is released and the time the processors in the node $A$  learn the secret, the adversary can not selectively decide to prevent the secret from being learned by taking over the processors which know the secret.  As all the secrets are sent down together to all the descendants of $A$, to prevent learning of a secret at $A$, the adversary must prevent a majority of level 1 nodes from hearing from $A$. This would require taking over enough nodes so that half the paths from $A$ to the leaves have at least one bad node in them, in which case we would view the election as bad and all arrays contending in it  bad.  While the adversary does have the ability to selectively make bad an entire election, this does not significantly affect the number of good arrays, since with high probability all elections return a representative fraction of good arrays.  See  Section \ref{election}. Nor does it affect the random bits which are later to be revealed, as these remain hidden from the adversary.  See  Section \ref{s:correct-comm}.

We  now lower  bound the fraction of arrays which remain which are good.

\begin{lemma}
At least a $2/3- 7 \ell /\log n  $ fraction of  winning arrays are good on every level $\ell$, that is, they are generated by good processors and
they are known by $1-O(1/\log n)$ fraction of good processors in their election node. 
In particular, the protocol can be used to generate a sequence of random words, of length $r=wq$ of which a $2/3 + \epsilon - 5 /\log \log  n$ fraction are random and known to $1-1/\log n$ fraction of good processors. 
\end{lemma}

\begin{proof}
With high probability, each good election causes an increase of  no more than a $1/\log n$ in the fraction of bad arrays, unless
there are too few ($<  c \log^3 n$) good arrays participating.   But the latter cannot happen too often.  Let $r$ be the number of contestants.
 If there are $f$ bad arrays overall on level $\ell$ then the total number of such lop-sided elections  is less than $f/(r-c \log^ 3 n)$. A representative set
 of winners would have $c \log^3 n/ r$ fraction of good arrays. Since the number of candidates $r=wq > \log^7$, the fraction of good arrays lost this way is less  than $1/\log^4 n$.
   So in total, the fraction of arrays which are good decreases by no more than $2/\log  n$ on a given level because of good election results.

 We now examine the effect of bad nodes. Each node makes bad any paths that run through it.  A fraction of $1/\log n$ level 1 nodes are bad. In the worst case, all arrays that pass through bad nodes and bad elections are good. Hence  a $1/\log n$ fraction of  bad nodes may eliminate a
$(3/2)\log n$ fraction of good arrays on level 1 and be responsible for making a $2/\log n$ fraction of elections bad in any level  by making bad half the paths of those elections, thus eliminating an additional fraction of $1/\log n$ good arrays.  On level 2, an additional $(5/2)/\log n$ fraction of elections may be made bad by the bad nodes in that level and so on.  Each bad election eliminates all good arrays which pass through it. Note that a bad election does not make additional paths bad, as information and secrets can still be passed through a good node that holds a bad election. 

 Initially a $2/3 + \epsilon $ fraction of the arrays are good. Assuming this is true,  the bad elections  may eliminate no more than $(5/2)/\log n$ fraction of good arrays.  Thus the fraction of elections which become bad on level $\ell^*$  is no more than a total of less than $5 \ell^*/\log n=5/\log \log n$.  Hence the fraction of good arrays at the root node is as stated. 
\end{proof}  

\section{Almost Everywhere to Everywhere}\label{s:AE2E}
We call a processor {\it knowledgeable} if it is good and agrees on a message $m$.   Otherwise, if it is good, it is {\it confused}.  We assume that almost all , i.e., $(1/2+\epsilon )n$,  of the processors are  knowledgeable and can come to agreement on a random number  $k$ in $[1,...,\sqrt{n}]$. We assume private channels.  Here is the protocol.

\begin{algorithm}

\caption{Almost Everywhere To Everywhere with Global Coin}

\begin{enumerate}

\item
Each processor $p$ does the following in parallel:  \\
Randomly pick  $i\in [1,2.,, ,\sqrt{n} ]$ and $j \in [1,...,n]$,  and send  a {\it request label} $i$ to processor $j$.
\item \label{random}
Almost all good processors agree on a random number  $k$ in $[1,...,\sqrt{n}]$.

\item \label{received}
For each processor $p$, if $p$ receives request label $i$ from $q$  and $i=k$ then  
if $p$ has not received more than $ \sqrt{n}\log n$ such messages (it is not {\it overloaded}) ,  
$p$ returns  a message to $q$, 
\item \label{majority}
Let  $k_i$ be the number of messages returned to $p$ by processors sent the request label $i$. Let $i_{max}$ be an $i$ such that  $k_i \geq k_j  \hbox{ for all j}\} \}$. If  the same message $m$  is  returned by $ (1/2 + 3\epsilon/8) a \log n  $  processors which  were sent the request label $i_{max}$ then $p$ decides  $m$.

\end{enumerate}
\end{algorithm}

\subsection{Proof of correctness}

\begin{lemma}\label{singleLoop}
Assume at the start of the protocol $n/2 + \epsilon n$ good processors agree on a message $M$ and can generate a  random bit.  Let $c$ be any constant > 0. Then after a single execution of the loop:
\begin{enumerate}
\item
With probability $4/(\epsilon \log n)-1/n^c$, this protocol results in agreement on $M$. 
\item
With probability $1-1/n^c$, every processor either agrees on $M$ or is undecided.
\end{enumerate}
\end{lemma}

To prove Lemma \ref{singleLoop} we first prove two other lemmas.
 
 \begin{lemma} \label{requests} Suppose there are $(1/2+ \epsilon) n $ knowledgeable processors.  W.h.p., for any one loop, for every processor $p$ and every request label $i$,
 at least $A=(1/2+ \epsilon/2) a \log n$ processors which are sent $i$ by $p$ are knowledgeable and fewer than $B=(1/2 - \epsilon/2)a \log n$ processors which are sent  $i$ by $p$ are corrupt or confused.
 
 \end{lemma}
 \begin{proof}
 Since there are private channels, the adversary does not know $p$'s requests other than those sent to bad processors .  Hence the choice of the set  of processors which are not knowledgeable is independent of the queries,  and each  event  consisting of a processor querying a knowledgeable processor is an independent random variable. 
 
 Let $X$ be the number of knowledgeable processors sent a value $i$ by processor $p$. $E[X]= a\log n (1/2 + \epsilon)$.
 Since $X$ is the sum of independent random variables we use
 Chernoff bounds:  $Pr[X < (1- \epsilon/2) E[X] \leq e^{(\epsilon^2/8)( a\log n (1/2 + \epsilon)}\leq n^{a\epsilon^2/16)}$ which is less than 
 $n^{-2c}$ for $a= 32c/\epsilon^2$.
 
Taking a union bound over all $i$ and  processors $p$,  for all $X$,  $Pr(X<   (1- \epsilon/2) E[X] )$
is less than $\sqrt{n}(n)n^{-2c} < 1/n^{-c}$. The second part of Lemma \ref{requests} is shown similarly.
\end{proof}

Lemma \ref{requests} immediately implies statement (2) of Lemma \ref{singleLoop}.

We now show Lemma \ref{singleLoop} (1). A knowledgeable processor $p$ which is sent $i=k$ will respond unless overloaded. Each processor can receive no more than $n-1$ requests, or the sender is evidently corrupt.
Then there can be no more than  $\sqrt{n}/\log n$ values of  $i$ for which there are more than  $\sqrt{n} \log  n$ requests labelled $i$. Then we claim:

 \begin{lemma}\label{overload}
The probability that more than $\epsilon n/4 $ knowledgeable processors are overloaded is less than $4/(\epsilon \log n)$.
\end{lemma}

\begin{proof} 
We call a value $i$ for a processor overloaded if $\sqrt{n} \log n$ request labels equal $i$.
A processor is only overloaded if $k=i$ and $i$ is overloaded.  Since $k$ is randomly chosen, each processor has 
at most a $1/\log n$ chance of being overloaded.  Let $X$ be the number of overloaded knowledgeable processors and
$Y$ be the number of knowledgeable processors.
Then $E[X]= Y/log n$. Using Markov's Inequality, $Pr[X \geq Y(\epsilon/4)] < (Y/\log n)/( Y \epsilon/4)=4/(\epsilon \log n)$. 
\end{proof}
Similar to the argument above, because the adversary does not know the requests and request labels of the requests 
sent to knowledgeable processors, the event sof choosing knowledgeable  processors which are not overloaded are independent random variables and Chernoff bounds apply. With probability $4/(\epsilon \log n)$, there are $(1/2 + 3\epsilon/4 n)$ knowledgeable processors which are not overloaded.
Setting $\epsilon$ to $3\epsilon/4$ in Lemma~ \ref{requests},  we  have w.h.p., for every processor and request label $i$ that $A=(1/2+ 3\epsilon/8) a \log n$ processor
and $B=(1/2 - 3\epsilon/8)a \log n$.  Therefore, with probability $4 / (\epsilon \log  n) -1/n^c$, one loop of this protocol results in agreement on $M$.   As each repetitions of the loop are independent, the probability that they all fail is the product of their individual failure
probabilities, implying the following.
 
\begin{lemma}
Repeating the protocol $X=(c/3) \epsilon  \ln n $ times,  the probability that all processors agree on $M$ and no processor outputs a different message is $1-1/n^c$. 

\end{lemma}

\section{Everywhere Byzantine Agreement}
We run the Almost Everywhere Agreement  protocol modified as in Section \ref{GCS}  to solve the global coin subsequence problem, i.e., it generates a polylogarithmic length sequence containing a subsequence of  $c\log n$ bits  random numbers  generated uniformly and independently at random which are
known to $1-1/\log n$ processors and are in the range $[1,...,\sqrt n]$.
At each step below, $GenerateSecretNumber(i)$ generates the $i^{th}$ number in the sequence.  Since the number of  good
random numbers is greater than $ c  \ln n $, the protocol is successful with probability $1-1/n^c$. 

\begin{algorithm}
\caption{Everywhere Byzantine Agreement}

\begin{enumerate}
\item
Run $Almost\_Everywhere\_Byzantine\_ Agreement$  to come to almost everywhere consensus on a bit;
\item
For $i=1$ to  $wq$ do
\begin{enumerate}
\item
$R \leftarrow $ GenerateSecretNumber$(i)$
\item
Run $AlmostEverywhereToEverywhere(R)$
\end{enumerate}
\end{enumerate}
\end{algorithm}

Finally,  it is easy to see that each execution of the $AlmostEverywhereToEverywhere$ takes $\tilde{O}(\sqrt{n})$ bits per
processor, which dominates the cost per processor. As there are polylogarithmic number ($wq$) of rounds, the communication cost  of $Everywhere Byzantine Agreement$ per processor remains  $\tilde{O}(\sqrt{n})$ bits while the time is polylogarithmic. 

\section{Conclusion}
We have described an algorithm that solves the Byzantine agreement problem with each processor sending only $\tilde{O}(\sqrt{n})$ bits.  Our algorithm succeeds against an adaptive, rushing adversary in the synchronous communication model.  It assumes private communication channels but makes no other cryptographic assumptions.  Our algorithm succeeds with high probability and has latency that is polylogarithmic in $n$.  Several important problems remain including the following:  Can we use $o(\sqrt{n})$ bits per processor, or alternatively prove that $\Omega(\sqrt{n})$ bits are necessary for agreement against an adaptive adversary?  Can we adapt our results to the asynchronous communication model?  Can we use the ideas in this paper to perform scalable, secure multi-party computation for other functions?  Finally, can the techniques in this paper be used to create a practical Byzantine agreement algorithm for real-world, large networks?

\pagebreak

\bibliography{security}
\bibliographystyle{plain}

\appendix

\section{Appendix}

\subsection{Proof of Lemma~\ref{l:analysis}} \label{s:pf-analysis}

\begin{proof}
We first analyze running time. There are  $q$ rounds in the first execution of  Step 2(c)  and $q*w$ rounds on the second and later executions of
Step 2(c) and Step 2.  Each round takes the time needed to traverse up and down to the node running the election or $O(\ell^*)$. 
The total running time is $O(\ell^*(q(w+1) ))$.

We now consider the number of bits communicated per processor. We note that each processor appears in all node only polylogarithmic number of times. Hence it suffices to bound the cost per appearance of processor in a node to get  a $\tilde{O}$ result. 
Step 1 requires each processor to  generate  $q + (\ell^*-1)wq +1$ words. Each share takes the same number of bits as the secret shared,  and there are
$k_1$ shares. When a processor in a level 1 node receives its share,  it  shares it with its parent node via $q \log^3 n$ uplinks, for a total of ${O}((q  \log^3 n + k_1)(q + \ell^*wq + 1))$ words.  sent by each processor. 

In the first execution of Step 2 (a) and (b) , every processor in every node $C$ on level 2 has  2-shares of the $q$ first blocks from its children, each containing $q$ words. These are  sent from every node $C$ on level 2 down its uplinks
to processors in all its  level 1 children,  so that each processor in  $C$ sends down $ {O}(q^2 d_m)$ words in total. Here, $d_m $ is the maximum number of uplinks from a single child that a processor in a node is incident to. The 1-shares are reconstructed from the 2-shares and then the 1-shares are shared with the other processors in the level 1 node, with each level 1 processor sending $q^2$ words in total. 

 Step 2(c) requires  shares of arrays from $w$ winners, or a total of $(\ell^*-1)w$ blocks to be sent secretly to the next level.  Each block has size $qw$ and is shared among ${O}q \log^3 n)$ processors, where $d=q \log^3 n$ is the number of a processor's uplinks. 
 for a total of $\tilde{O}(\ell^* (wq)^2 ))$ words sent. 
 
In the second execution of Step 2(b) and (c) all shares of all $F$ blocks of all $wq$ candidates are sent down from $C$ at level 3.  Each processor has received 
2-shares from  ${O}(d^2_m wq)$ candidates, hence it sends down $\tilde{O}(d^2_m wq)$ shares of blocks of size $wq$ or  a total of ${O}((d^2_m (wq)^2)$
words.
On level $1$ the 2-shares are converted to 1-shares and each 1-share is sent to $k_1$ processors, for a total of $(wq)^2$ words sent to $k_1$ processors or
$O(k_1 (wq)^2)$ words sent. 
The processors in level 1 nodes each determine $(wq)^2$ numbers which they communicate back to the level 3 nodes via the $\ell-links$ to their neighbors in $C$.
Since each level 1 node is incident to  $\tilde{q })$ $\ell-links$, for any level $\ell$ node, the cost of this is $\tilde{O}( (w^2 q^3)$ words.

In the third execution of Step 2(b) and (c), again all shares of $F$ blocks of all $wq$ candidates are sent down. Each has size $wq$ and the analysis is similar, except for one item. As the levels increase to level $\ell$, the number of candidates from which a candidate has received $\ell$-shares
increases by a $d_m$ factor. Hence, each processor at level $\ell$ sends down ${O}((d^{\ell}_m (wq)^2)$ words. 

Step 3 is dominated by Step 2.

The total number of bits sent per processor is the number of times a processor appears in a node on any level times the number of levels times
the costs described above. These additional factors  for appearances and levels are subsumed by the $\tilde{}$ notation as they are polylog .
That is, the cost is determined by summing up the above amounts with the exception of one term which increases per level, that is, $\sum_{\ell} (d^{\ell}_m (wq)^2))$. Hence the cost is 

 $$ \tilde{O}((q + k_1)(q + \ell^* wq) + \ell^*(wq)^2  +  k_1 (wg)^2 + w^2 q^3 +
\sum_{\ell} (d^{\ell}_m (wq)^2))$$

$$=\tilde{O}((w^2 q^3 + 2d_m^{\ell^*} (wq)^2))$$.

Since $w=O(\log^3 n)$, $\ell^*= \log (n/k_1)/\log q$,  $d_m = c' \log^4 n$,   and $k_1$ to $\log^3 n$, then setting  $q=(\log n)^{\delta}$, $\delta >4$,   we have that  the total cost is dominated by the last term and $$ d_m^{\ell^*} (wq)^2=
(c' \log^4 n)^{ \log (n/k_1)/\log q} (wq)^2$$$$= 2^{\log c' + 4 \log \log n ) (\log (n/k_1)/\log q} (c\log^3 n q)^2= \tilde{O}(n^{4/\delta})$$.

I.e., there is a running time of $O(\log(n/\log^3 n)/ \log \log n) (\log n)^{\delta}  \log^3 n)=O((\log n)^{4+\delta} )$ and a bit complexity per processor of
$\tilde{O}(n^{4/\delta})$. 
\end{proof}

\subsection{Almost Everywhere Byzantine Agreement (AEBA) with Unreliable Global Coins}\label{AEBACC}

\begin{algorithm}
\caption{AEBA with Unreliable Coins} \label{alg:aebasparse}
Set $vote \leftarrow b_{i}$; For each round do the following:
\begin{enumerate*}
\item Send $vote$ to all neighbors in $G$;
      \label {step_a}
\item Collect votes from neighbors in $G$;
      \label {step_b}
\item $maj \leftarrow$ majority bit among votes received;
\item $fraction \leftarrow$ fraction of votes received for $maj$;
\item $coin \leftarrow$ result of call to algorithm GetGlobalCoin;
\item If $fraction \geq (1-\epsilon_{0}) (2/3 + \epsilon/2)$ then $vote \leftarrow maj$
\item else
\begin{enumerate}
\item If $coin$ = ``heads'', then $vote \leftarrow 1$, else $vote \leftarrow 0$;
\end{enumerate}
\end{enumerate*}

At the end of all rounds, commit to $vote$ as the output bit;
\end{algorithm}

\subsection{Analysis}

We assume here that the fraction of bad processors is no more than $1/3 + \epsilon$ for some fixed $\epsilon > 0$.  For a processor $v$, let $N(v)$ be the set of neighbors of $v$ in the sparse graph $G$ and let $n$ be the number of nodes in this graph.  We say that a call to $GetGlobalCoin$ \emph{succeeds}, when it selects a bit $b \in \{0,1\}$ with uniform probability and independently from all past events, and that all but $O(n/ \log n)$ processors learn $b$.

\begin{theorem} 
Assume there are at least $r$ rounds in Algorithm~\ref{alg:aebasparse} where the call to $GetGlobalCoin$ succeeds.  Let $C_{1}$ and $C_{2}$ be any positive constants and $k$ depend only on $C_{1}$ and $C_{2}$.  Then, at the end of the algorithm, for any positive constants $C_{1}$ and $C_{2}$ with probability at least $1-e^{-C_{1}n} + 1/2^{r}$, all but $C_{2}n / \log n$ of the good processors commit to the same vote $b$, where $b$ was the input of at least one good processor.   This occurs provided that the graph $G$ is a random $k \log n$ regular graph.
\end{theorem}

Before proving this theorem we establish the following lemmas.

For a fixed round, let $b' \in \{ 0, 1\}$ be the bit that the majority of good processors vote for in that round.  Let $S'$ be the set of good processors that will vote for $b'$ and let $f' = |S'|/n$.  Let $\epsilon_{0}$ be a fixed positive constant to be determined later.  We call a processor \emph{informed} for the round if the fraction value for that processor obeys the following inequalities:
$$
(1-\epsilon_{0}) f'  \leq fraction \leq (1+\epsilon_{0}) (f' + 1/3 - \epsilon)
$$

\begin{lemma}
For any fixed $C_{1}$ and $C_{2}$, with probability at least $1- e^{-C_{1}n}$, in any given round of Algorithm~\ref{alg:aebasparse}, all but $C_{2}n/\log n$ of the good processors are informed, for $G$ a $k \log n$ regular graph where $k$ depends only on $C_{1}$, $C_{2}$ and $\epsilon_{0}$. 
\end{lemma}

\begin{proof}
Fix the set $S'$, we know that $S'$ is of size at least $1/3(n + \epsilon)$ since at least half of the good processors must vote for the majority bit.  Let $f' = |S'|/n$. We will also fix a set $T_{\ell}$ which consists of all the processors that have $fraction < (1-\epsilon_{0} f')$.  We will first show that the probability that $T_{\ell}$ is of size $C n/ 2\log n$ for some constant $C$ is very small for fixed $S'$ and $T_{\ell}$, and will then show, via a union bound, that with high probability, for any $S'$ there is no set of $C n/ 2\log n$ processors with $fraction < (1-\epsilon_{0} f')$.  Finally, we will use a similar technique to show that with high probability, no more than $C n/ 2\log n$ processors have $fraction > (1+\epsilon_{0}) (f' + 1/3)$.  This will complete the proof.

To begin, we fix the set $S'$ of size at least $1/3(n + \epsilon)$ and fix $T_{\ell}$ of size $C n/ (2\log n)$.  Let $\xi(S',T)$ be the event that all processors in $T_{\ell}$ have $fraction < (1-\epsilon_{0}) f'$.  Let $X$ be the number of edges from $S'$ to $T_{\ell}$.  Since the graph $G$ is $k \log n$ regular, we know that $Pr(\xi(S',T)) = Pr(X < (1-\epsilon_{0}) f' |T| k \lg n)$.  We will find an upper-bound on the latter probability by using a random variable $Y$ that gives the number of edges from $S'$ to $T_{\ell}$ if the graph $G$ were generated by having $k \log n$ edges from each vertex with endpoint selected uniformly at random.  In particular, $X$ is the number of edges between the two sets if $G$ is a random regular graph, and $Y$ is the number of edges if $G$ is a graph where the out degree of each node is the same but the in-degrees may differ.  We know that $Pr(X < (1-\epsilon_{0}) f' |T| k \lg n) \leq Pr(Y < (1-\epsilon_{0}) f' |T| k \lg n)$ since the model for generating $X$ assumes sampling without replacement and that for $Y$ assumes sampling with replacement.  We will thus bound the probability of deviation for $Y$.  Note that $E(Y) = f' |T| k \lg n$, and so by Chernoff bounds, we can say that 

\begin{eqnarray*}
Pr (Y < (1-\epsilon_{0}) f' |T| k \lg n) & \leq & e^{-(k/4)(\epsilon_{0})^{2} f' C n} \\
& = & e^{-(k/12)(\epsilon_{0})^{2} C n}
\end{eqnarray*}

Where the last step holds since $f' \geq 1/3$.  Let $\xi$ be the union of events $\xi(S',T)$ for all possible values of $S'$ and $T$.  Then we know by union bounds that

\begin{eqnarray*}
Pr (\xi) &= & \sum_{S',T} Pr(\xi(S',T)) \\
& \leq & 2^{n}2^{n}e^{-(k/12)(\epsilon_{0})^{2} C n} \\
& \leq & e^{-C' n}
\end{eqnarray*}

Where the last equation holds for any constant $C'$, provided that $k$ is sufficiently large but depends only on the constants $C$ and $\epsilon_{0}$.  We have thus shown that with high probability, the number of processors with $fraction < (1-\epsilon_{0} f')$ is no more than $C n/ 2\log n$.  

By a similar analysis, letting $S''$ consist of the union of $S'$ and the set of bad processors, we can show that with high probability, the number of processors with $fraction > (1+ \epsilon_{0} f' + 1/3 - \epsilon)$ is no more than $C n/ 2\log n$.  These two results together establish that with high probability, in any round of the algorithm, all but $Cn/ log n$ processors are informed for any constant $C$, provided that $k$ is chosen sufficiently large with respect to $C$ and $\epsilon_{0}$.
\end{proof}

The following Lemma establishes validity (that the output bit will be the same as the input bit of one good processor) and will also be helpful in establishing consistency (that all but $Cn/\log n$ good processors will output the same bit).

\begin{lemma} \label{l:validity}
If in any given round all but $Cn/\log n$ good processors vote for the same value $b'$, for some constant $C$, then for every remaining round, all but $Cn/ \log n$ good processors will vote for $b'$.
\end{lemma}

\begin{proof}
We will show that if all but $Cn/\log n$ good processors vote for the same value $b'$ in some round $i$, then in round $i+1$, all but $Cn/ \log n$ good processors will vote for $b'$.  Consider what happens after the votes are received in round $r$.  We know that for this round, $f' \geq 2/3 + \epsilon - C/ \log n \geq 2/3 + \epsilon/2$ for $n$ sufficiently large.  Thus, every informed processor in that round will have $fraction \geq (1-\epsilon_{0}) f' \geq (1-\epsilon_{0}) (2/3 + \epsilon/2)$, and so every informed processor will set its vote value, at the end of the round, to $b'$.  It follows that all the processors that were informed in round $i$ will vote for $b'$ in round $i+1$.  Note that this result holds irrespective of the outcome of GetGlobalCoin for the round, even including the case where different processors receive different outcomes from that subroutine.
\end{proof}

\begin{lemma} \label{l:term}
If the call to GetGlobalCoin succeeds in some round (i.e. the same unbiased coin toss is returned to all but $O(n/\log n)$ good players), then with probability at least $1/2$, at the end of that round, all but $O(n/ \log n)$ good processors will have a vote value equal to the same bit.
\end{lemma}

\begin{proof}
Fix a round where the call to GetGlobalCoin succeeds.  There are two main cases\\
\noindent
Case 1: No informed processor has $fraction \geq (1-\epsilon_{0}) (2/3 + \epsilon/2)$.  In this case, at the end of the round, with probability $1$, all but $Cn/\log n$ processors will set their vote to the same bit.
\noindent
Case 2: At least one informed processor has $fraction \geq (1-\epsilon_{0}) (2/3 + \epsilon/2)$.  We first show that in this case, all informed processors that have 
$fraction \geq (1-\epsilon_{0}) (2/3 + \epsilon/2)$ will set their vote to the same value at the end of the round.  We show this by contradiction.  Assume there are two processors, $x$ and $y$, where $fraction_{x}$ ($fraction_{y}$) are the fraction values of $x$ ($y$), such that both $fraction_{x}$ and $fraction_{y}$ are greater than or equal to $(1-\epsilon_{0}) (2/3 + \epsilon/2)$, and $x$ sets its vote to $0$ at the end of the round, while $y$ sets its vote to $1$.

Let $f'_{0}$ ($f'_{1}$) be the fraction of good processors that vote for $0$ ($1$) during the round.  Then we have that $fraction_{x} \geq (1-\epsilon_{0}) (2/3 + \epsilon/2)$.  By the definition of informed, we also know that $fraction_{x} \leq (1+\epsilon_{0}) (f'_{0} + 1/3 - \epsilon)$.  This implies that
$$ (1-\epsilon_{0})(2/3 + \epsilon/2) \leq (1+\epsilon_{0})(f'_{0} + 1/3 - \epsilon). $$  Isolating $f'_{0}$ in this inequality, we get that
$$f'_{0} \geq \frac{1/3 + 3/2 \epsilon - \epsilon_{0} - (3/2) \epsilon \epsilon_{0}}{1+\epsilon_{0}}.$$

A similar analysis for $fraction_{b}$ implies that
$$ f'_{1} \geq \frac{1/3 + 3/2 \epsilon - \epsilon_{0} - (3/2) \epsilon \epsilon_{0}}{1+\epsilon_{0}}.$$

But then, for $\epsilon_{0}$ sufficiently small, we have $f'_{0} + f'_{1} > 2/3 + \epsilon$, which is a contradiction.  

Now, let $b'$ be the value that all good and informed processors with $fraction \geq (1-\epsilon_{0}) (2/3 + \epsilon/2)$ set their value to at the end of the round.   With probability $1/2$, the outcome of the GetGlobalCoin is equal to $b$ and in this case, all but $O(n/ \log n)$ informed processors will set their vote value to the same bit $b$ at the end of the round.
\end{proof}

We can now prove Theorem~\ref{t:aebasparse}.\\
\begin{proof}
Lemma~\ref{l:validity} establishes validity: if all processor initially start with the same input bit, then all but $C_{2}n/ \log n$ of the processors will eventually commit to that bit, with probability at least $1-e^{-C_{1}n}$.  Lemmas~\ref{l:term} and~\ref{l:validity} together establish that the probability of having a round in which all but $C_{2}n/ \log n$ processors come to agreement (and after which all but $C_{2}n/ \log n$ processors  will stay in agreement) is at least $1-2^{r}$ where $r$ is the number of rounds in which $GetGlobalCoin$ succeeds.  A simple union bound on the probabilities of error then establishes the result of the theorem.
\end{proof}

\end{document}